\documentclass[11pt]{article}

\usepackage{palatino}
\usepackage{mathpazo}
\usepackage{stmaryrd}

\usepackage{hyperref}
\hypersetup{pdfpagemode=UseNone}

\usepackage{amsfonts}
\usepackage{amssymb}
\usepackage{amsmath}
\usepackage{latexsym}
\usepackage{amsthm}

\newtheorem{theorem}{Theorem}

\newtheorem{prop}[theorem]{Proposition}
\newtheorem{cor}[theorem]{Corollary}
\theoremstyle{definition}

\newtheorem{problem}[theorem]{Problem}

\newcommand{\tinyspace}{\mspace{1mu}}
\newcommand{\microspace}{\mspace{0.5mu}}
\newcommand{\op}[1]{\operatorname{#1}}
\newcommand{\norm}[1]{\left\lVert\tinyspace#1\tinyspace\right\rVert}

\newcommand{\abs}[1]{\left\lvert\tinyspace #1 \tinyspace\right\rvert}

\newcommand{\dnorm}[1]{\norm{#1}_{\diamond}}

\newcommand{\tr}{\operatorname{Tr}}

\newcommand{\ip}[2]{\left\langle #1 , #2\right\rangle}
\newcommand{\sip}[2]{\langle #1 , #2\rangle}
\newcommand{\triplenorm}[1]{
  \left|\!\microspace\left|\!\microspace\left| #1 
  \right|\!\microspace\right|\!\microspace\right|}
\newcommand{\fid}{\operatorname{F}}
\newcommand{\setft}[1]{\mathrm{#1}}
\newcommand{\lin}[1]{\setft{L}\left(#1\right)}
\newcommand{\density}[1]{\setft{D}\left(#1\right)}
\newcommand{\unitary}[1]{\setft{U}\left(#1\right)}
\newcommand{\trans}[1]{\setft{T}\left(#1\right)}
\newcommand{\herm}[1]{\setft{Herm}\left(#1\right)}
\newcommand{\pos}[1]{\setft{Pos}\left(#1\right)}
\newcommand{\pd}[1]{\setft{Pd}\left(#1\right)}
\newcommand{\sphere}[1]{\mathcal{S}\!\left(#1\right)}

\newenvironment{mylist}[1]{\begin{list}{}{
	\setlength{\leftmargin}{#1}
	\setlength{\rightmargin}{0mm}
	\setlength{\labelsep}{2mm}
	\setlength{\labelwidth}{8mm}
	\setlength{\itemsep}{0mm}}}
	{\end{list}}

\def\complex{\mathbb{C}}

\def\({\left(}
\def\){\right)}
\def\I{\mathbb{1}}
\def\X{\mathcal{X}}
\def\Y{\mathcal{Y}}
\def\Z{\mathcal{Z}}
\def\W{\mathcal{W}}
\def\A{\mathcal{A}}
\def\B{\mathcal{B}}

\def\U{\mathcal{U}}

\def\R{\mathcal{R}}

\def\Q{\mathcal{Q}}

\begin{document}

\title{\bf Semidefinite programs for completely bounded norms}

\author{John Watrous\\[1mm]
  {\it\small Institute for Quantum Computing and School of Computer
    Science}\\ 
  {\it \small University of Waterloo, Waterloo, Ontario, Canada}
}

\date{April 15, 2009}

\maketitle

\begin{abstract}
  The completely bounded trace and spectral norms in finite dimensions
  are shown to be expressible by semidefinite programs.
  This provides an efficient method by which these norms may be both
  calculated and verified, and gives alternate proofs of some known
  facts about them.
\end{abstract}

\section{Introduction}

Linear mappings from one space of operators to another, which are
often called {\it super-operators}, play an important role in quantum
information theory.
Quantum channels in particular, which model general discrete-time
changes in quantum systems, are represented by super-operators that
act on operators on finite-dimensional complex vector spaces.

It is natural to consider {\it distances} between quantum channels, so
as to quantify the similarity with which they act on quantum states.
One way to define such a notion is to define a suitable {\it norm}
on the space of super-operators in which channels of a given size are
represented.
Then, the distance between two channels is defined as the norm of
their difference.
A natural question that arises is: {\it what norms give rise to the
most physically meaningful notions of distance?}
As is argued in \cite{GilchristLN05}, the answer to this question may
depend on the problem at hand---but perhaps the most natural and
widely applicable choice within quantum information theory is the
{\it completely bounded trace norm}, also known as the
{\it diamond norm}.
This norm was first used in the setting of quantum information by
Kitaev \cite{Kitaev97}, who used it mainly as a tool in studying
quantum error correction and fault-tolerance.
It is equivalent, up to taking the adjoint of a super-operator,
to its spectral norm variant, which is usually known simply as the 
{\it completely bounded norm}.
The completely bounded norm, as well as variants that include the
completely bounded trace norm, have been studied in operator
theory for many years.
(See \cite{Paulsen02} for historical comments and further details.)

The definition of these completely bounded norms, which can be found
in the section following this introduction, may seem unusual at first
glance.
It turns out, however, that they are quite natural and satisfy
many remarkable properties.
They are, in particular, much easier to reason about and to work with
than the seemingly simpler super-operator norms that are induced by
the trace norm and spectral norm, primarily because the completely
bounded norms respect the structure of tensor products while the
induced norms do not.
The physical importance of this property, within the setting of
quantum information theory, has been discussed in several sources
\cite{
  Kitaev97,
  AharonovKN98,
  ChildsPR00,
  D'ArianoPP01,
  Acin01,
  GilchristLN05,
  RosgenW05,
  Sacchi05b,
  Sacchi05,
  Rosgen08,
  Watrous08,
  PianiW09}.
Additional references that highlight the properties and uses of
completely bounded norms in quantum information include
\cite{DevetakJKR06,Jenvccova06,Perez-GarciaWPVJ08}.

One obvious question that comes to mind about the completely bounded
trace and spectral norms is: {\it can they be efficiently computed?}
Unlike the norms of operators that are most typically encountered in
quantum information theory, which are trivially computable from
spectral or singular-value decompositions, the computation of
completely bounded norms is not known to be straightforward.
To the author's knowledge there are only two papers
written prior to this one, namely \cite{Zarikian06} and
\cite{JohnstonKP09}, that present methods to compute the completely
bounded spectral or trace norm of a given super-operator.
Both papers describe iterative methods, and analyze the complexity of
each iteration of these methods, but do not analyze their rates of
convergence.
So, although these papers may describe potentially efficient methods,
they do not include complete proofs of their efficiency.

The purpose of this paper is to explain how the completely bounded
trace norm of a given super-operator (and therefore its completely
bounded spectral norm as well) can be expressed as the optimal value
of a semidefinite program whose size is polynomial in the dimension of
the spaces on which the super-operator acts.
Using known polynomial-time algorithms for solving semidefinite
programs, one obtains a provably efficient algorithm (meaning a
deterministic polynomial time algorithm) for calculating these norms.
This approach also has the obvious practical advantage that it is more
easily implemented through the use of existing semidefinite programming
optimization libraries, and allows one to take advantage of the
extensive work that has been done to solve semidefinite programs
efficiently and accurately.
Moreover, through semidefinite programming duality, one obtains a
means by which a certificate of the value of the completely bounded
trace or spectral norm of a given super-operator may be quickly
verified.

In a recent paper written independently from this one, Ben-Aroya and
Ta-Shma \cite{Ben-AroyaT09} have found a different (but related) way
to efficiently compute the completely bounded trace norm of
super-operators using convex programming.

The essence of the semidefinite programming formulation of the
completely bounded trace norm that is described in this paper appears,
at least to some extent, in the paper \cite{KitaevW00}; although it
was not made explicit or considered in full generality therein.
The present paper aims to present this formulation explicitly and
without any discussion of the {\it quantum interactive proof system}
model of computation, which is the primary focus of \cite{KitaevW00}.
A second semidefinite programming formulation of the completely
bounded trace norm is also presented, based on the
{\it competitive quantum game} framework of \cite{GutoskiW07}.
This formulation is slightly simpler, but is valid only for
super-operators that are the difference between two quantum
channels---which happens to be an important special case in quantum
information.

Semidefinite programming is useful not only as a computational tool,
but as an analytic tool as well.
The last section of this paper gives two examples along these lines
that are derived from the more general semidefinite programming
formulation of the completely bounded trace norm.
The first example concerns an alternate characterization of the completely
bounded trace norm and the second illustrates a precise sense in which
two known characterizations of the fidelity function (given by Uhlmann's
Theorem and Alberti's Theorem) are dual statements to one another.

\vspace{-3mm}

\section{Background}

The two subsections that follow aim to provide the reader with an
account of the background knowledge assumed in the remainder of the
paper.
The first subsection discusses well-known concepts from
finite-dimensional operator or matrix theory, and is mainly intended
to make clear the notation and terminology that is used later.
It also includes the definitions of the completely bounded norms that
are the focus of this paper.
The second subsection discusses semidefinite programming, using a
form that is less common than the so-called {\it standard form}
of semidefinite programs, but that is equivalent and better suited to
the needs of this paper.

\subsection*{Operators and super-operators}

The scripted letters $\X$, $\Y$, $\Z$, and $\W$ will denote vector
spaces of the form $\complex^n$ for $n\geq 1$, whose elements are
identified with $n$-dimensional column vectors.
The $j$-th elementary unit vector in such a space is denoted $e_j$,
and the unit sphere in $\X$ (with respect to the Euclidean norm) is
denoted
\[
\sphere{\X} = \left\{ u\in\X\,:\,\norm{u} = 1\right\}.
\]

For $\X = \complex^n$ and $\Y = \complex^m$, the space consisting of
all linear operators of the form $A:\X\rightarrow\Y$ is denoted
$\lin{\X,\Y}$ and is identified with the set of $m\times n$ complex
matrices in the usual way.
The notation $\lin{\X}$ is shorthand for $\lin{\X,\X}$, and the space
$\X$ is identified with $\lin{\complex,\X}$ when necessary.
An inner product on $\lin{\X,\Y}$ is defined as 
$\ip{A}{B} = \tr(A^{\ast} B)$ for all $A,B\in\lin{\X,\Y}$, where
$A^{\ast}\in\lin{\Y,\X}$ denotes the adjoint (or conjugate transpose)
of $A$.
The identity operator on $\X$ is denoted $\I_{\X}$, and for each
choice of indices $i,j$ we write $E_{i,j} = e_i e_j^{\ast}$.

Three operator norms are discussed in this paper: the 
{\it trace norm}, {\it Frobenius norm}, and {\it spectral norm},
defined as
\[
\norm{A}_1 = \tr\sqrt{A^{\ast} A}\:,
\quad\quad
\norm{A}_2 = \sqrt{\ip{A}{A}}\:,
\quad\quad\text{and}\quad\quad
\norm{A}_{\infty} = \max_{u\in\sphere{\X}}
\norm{Au}\:,
\]
respectively, for each $A\in\lin{\X,\Y}$.
For every operator $A$ it holds that
\[
\norm{A}_{\infty} \leq \norm{A}_2 \leq \norm{A}_1.
\]

The following special types of operators are also discussed.
\begin{mylist}{\parindent}
\item[1.]
An operator $X\in\lin{\X}$ is {\it Hermitian} if $X = X^{\ast}$.
The set of such operators is denoted $\herm{\X}$.

\item[2.]
An operator $P\in\lin{\X}$ is {\it positive semidefinite} if it is
Hermitian and all of its eigenvalues are nonnegative.
The set of such operators is denoted $\pos{\X}$.
The notation $P\geq 0$ also indicates that $P$ is positive
semidefinite, and more generally the notations $X\leq Y$ and $Y\geq X$
indicate that $Y - X\geq 0$ for Hermitian operators $X$ and $Y$.

\item[3.]
An operator $P\in\lin{\X}$ is {\it positive definite} if it is
both positive semidefinite and invertible.
The set of such operators is denoted $\pd{\X}$.
The notation $P>0$ also indicates that $P$ is positive definite, and
the notations $X<Y$ and $Y>X$ indicate that $Y - X> 0$ for
Hermitian operators $X$ and $Y$.

\item[4.]
An operator $\rho\in\lin{\X}$ is a {\it density operator} if it is
both positive semidefinite and has trace equal to 1.
The set of such operators is denoted $\density{\X}$.

\item[5.]
An operator $U\in\lin{\X}$ is {\it unitary} if $U^{\ast} U = \I_{\X}$.
The set of such operators is denoted $\unitary{\X}$.
\end{mylist}

A {\it super-operator} is a linear mapping of the form
$\Phi:\lin{\X}\rightarrow\lin{\Y}$, and the space of all mappings 
of this form is denoted $\trans{\X,\Y}$.
The identity super-operator on $\lin{\X}$ is denoted $\I_{\lin{\X}}$.
The adjoint super-operator to $\Phi\in\trans{\X,\Y}$ is the unique
super-operator $\Phi^{\ast}\in\trans{\Y,\X}$ for which
$\ip{Y}{\Phi(X)}=\ip{\Phi^{\ast}(Y)}{X}$ for all $X\in\lin{\X}$ and
$Y\in\lin{\Y}$.
The following special types of super-operators are discussed.
\begin{mylist}{\parindent}
\item[1.]
  $\Phi\in\trans{\X,\Y}$ is 
  {\it Hermiticity-preserving} if $\Phi(X) \in \herm{\Y}$ for every
  $X\in\herm{\X}$.

\item[2.]
  $\Phi\in\trans{\X,\Y}$ is {\it completely positive}
  if it holds that
  \[
  (\Phi\otimes\I_{\lin{\W}})(P) \in \pos{\Y\otimes\W}
  \]
  for every choice of $\W = \complex^k$ and $P\in\pos{\X\otimes\W}$.
  
\item[3.]
  $\Phi\in\trans{\X,\Y}$ is {\it trace-preserving} if
  $\tr(\Phi(X)) = \tr(X)$
  for every $X\in\lin{\X}$.

\item[4.]
  $\Phi\in\trans{\X,\Y}$ is a {\it quantum channel} if it is both
  completely positive and trace-preserving.

\end{mylist}

The {\it Choi-Jamio{\l}kowski representation}
$J(\Phi)\in\lin{\Y\otimes\X}$ of a super-operator
$\Phi\in\trans{\X,\Y}$ is the operator defined as
\[
J(\Phi) = \sum_{1\leq i,j \leq n} \Phi(E_{i,j}) \otimes E_{i,j}
\]
(where this expression assumes $\X = \complex^n$).
The mapping $J$ is a linear bijection from $\trans{\X,\Y}$ to
$\lin{\Y\otimes\X}$.
The operator $J(\Phi)$, written as an $nm\times nm$ matrix, represents
one convenient way that a super-operator may be expressed in concrete
terms.

A pair of operators $(A,B)$ in $\lin{\X,\Y\otimes\Z}$ is a
{\it Stinespring pair} for $\Phi\in\trans{\X,\Y}$ if it holds that
\begin{equation} \label{eq:Stinespring}
\Phi(X) = \tr_{\Z}\! \( A X B^{\ast} \)
\end{equation}
for all $X\in\lin{\X}$, and an expression of the form
\eqref{eq:Stinespring} is called a {\it Stinespring representation} of
$\Phi$.
Stinespring representations exist for every super-operator, provided
the space $\Z$ has dimension at least $\op{rank}(J(\Phi))$.
It is straightforward to compute such a Stinespring pair $(A,B)$ from
the Choi-Jamio{\l}kowski representation of $\Phi\in\trans{\X,\Y}$: for
any expression
\[
J(\Phi) = \sum_{l = 1}^r u_l v_l^{\ast},
\]
it holds that
\[
A = \sum_{i,j,l} \ip{e_i\otimes e_j}{u_l} E_{i,j} \otimes e_l
\quad\quad\text{and}\quad\quad
B = \sum_{i,j,l} \ip{e_i\otimes e_j}{v_l} E_{i,j} \otimes e_l
\]
forms a Stinespring pair of $\Phi$.

For every super-operator $\Phi\in\trans{\X,\Y}$, one defines
the {\it induced super-operator norms}:
\begin{align*}
  \norm{\Phi}_1 & = \max\left\{\norm{\Phi(X)}_1 \,:\, X\in\lin{\X},\,\;
  \norm{X}_1\leq 1\right\},\\[2mm]
  \norm{\Phi}_{\infty} & = \max\left\{\norm{\Phi(X)}_{\infty} \,:\,
  X\in\lin{\X},\,\;\norm{X}_{\infty}\leq 1\right\},
\end{align*}
as well as {\it completely bounded} variants of these norms:
\[
\triplenorm{\Phi}_1 = \sup_{k\geq 1}
\norm{\Phi\otimes\I_{\lin{\complex^k}}}_1
\quad\quad\text{and}\quad\quad
\triplenorm{\Phi}_{\infty} =
\sup_{k\geq 1}
\norm{\Phi\otimes\I_{\lin{\complex^k}}}_{\infty}.
\]
As was done in the introduction, we will refer to
$\triplenorm{\Phi}_1$ as the
{\it completely bounded trace norm} and to
$\triplenorm{\Phi}_{\infty}$ as the
{\it completely bounded spectral norm}.
It is common that $\triplenorm{\Phi}_1$ is denoted $\dnorm{\Phi}$ and
called the {\it diamond norm}, and that $\triplenorm{\Phi}_{\infty}$
is denoted $\norm{\Phi}_{\text{cb}}$ and called simply the 
{\it completely bounded norm}.
It holds that
\[
\triplenorm{\Phi}_1 = \norm{\Phi\otimes\I_{\lin{\X}}}_1
\quad\quad\text{and}\quad\quad
\triplenorm{\Phi}_{\infty} =
\norm{\Phi\otimes\I_{\lin{\Y}}}_{\infty},
\]
and that $\triplenorm{\Phi}_1 = \triplenorm{\Phi^{\ast}}_{\infty}$ for
every $\Phi\in\trans{\X,\Y}$.
These norms are both multiplicative with respect to tensor
products, meaning that
\[
\triplenorm{\Phi\otimes\Psi}_1 = 
\triplenorm{\Phi}_1 \triplenorm{\Psi}_1
\quad\quad\text{and}\quad\quad
\triplenorm{\Phi\otimes\Psi}_{\infty} = 
\triplenorm{\Phi}_{\infty} \triplenorm{\Psi}_{\infty}
\]
for any choice of super-operators $\Phi$ and $\Psi$.

\subsection*{Semidefinite programming}

This section gives a brief overview of semidefinite programming, which
is discussed in greater detail in several sources (including
\cite{Alizadeh95,VandenbergheB96,Lovasz03,deKlerk02}, for instance).
The particular formulation that is described here is somewhat
different than the well-known {\it standard form} that is used by most
authors, but it is equivalent and more convenient for the purposes of
this paper.

A {\it semidefinite program} over $\X=\complex^n$ and
$\Y\in\complex^m$ is specified by a triple $(\Psi,A,B)$, where
\begin{mylist}{\parindent}
\item[1.] $\Psi\in\trans{\X,\Y}$ is a Hermiticity preserving
  super-operator, and
\item[2.] $A\in\herm{\X}$ and $B\in\herm{\Y}$ are Hermitian operators.
\end{mylist}
The following two optimization problems are associated with such a
semidefinite program:
\begin{center}
  \begin{minipage}{2in}
    \centerline{\underline{Primal problem}}\vspace{-7mm}
    \begin{align*}
      \text{maximize:}\quad & \ip{A}{X}\\
      \text{subject to:}\quad & \Psi(X) \leq B,\\
      & X\in\pos{\X}.
    \end{align*}
  \end{minipage}
  \hspace*{12mm}
  \begin{minipage}{2in}
    \centerline{\underline{Dual problem}}\vspace{-7mm}
    \begin{align*}
      \text{minimize:}\quad & \ip{B}{Y}\\
      \text{subject to:}\quad & \Psi^{\ast}(Y) \geq A,\\
      & Y\in\pos{\Y}.
    \end{align*}
  \end{minipage}
\end{center}

\noindent
With these problems in mind, one defines the {\it primal feasible} set
$\A$ and the {\it dual feasible} set $\B$ as 
\begin{align*}
\A & = \left\{X\in\pos{\X}\,:\,\Psi(X) \leq B\right\},\\
\B & = \left\{Y\in\pos{\Y}\,:\,\Psi^{\ast}(Y) \geq A\right\}.
\end{align*}
Operators $X\in\A$ and $Y\in\B$ are also said to be {\it primal feasible}
and {\it dual feasible}, respectively.
For the sake of the discussion of computational efficiency below, it
will be helpful to also define, for each $\varepsilon>0$, the sets
\begin{align*}
\A_{\varepsilon} & = \left\{X\in\pos{\X}\,:\,
X + H \in \A\;\text{for all $H\in\herm{\X}$ satisfying 
  $\norm{H}_2\leq\varepsilon$}\right\},\\
\B_{\varepsilon} & = \left\{Y\in\pos{\Y}\,:\,
Y + H \in \B\;\text{for all $H\in\herm{\Y}$ satisfying 
  $\norm{H}_2\leq\varepsilon$}\right\}.
\end{align*}
Intuitively speaking, $\A_{\varepsilon}$ contains primal feasible
operators that are not too close to the boundary of the primal
feasible set, and likewise for $\B_{\varepsilon}$.

The functions $X\mapsto\ip{A}{X}$ and $Y\mapsto\ip{B}{Y}$ are
called the primal and dual {\it objective functions}, and the
{\it optimal values} associated with the primal and dual problems
are defined as follows:
\[
\alpha = \sup_{X\in\A} \ip{A}{X}\quad\quad\text{and}\quad\quad
\beta = \inf_{Y\in\B} \ip{B}{Y}.
\]
(If it is the case that $\A = \varnothing$ or $\B = \varnothing$,
the above definitions are to be interpreted as $\alpha = -\infty$ and
$\beta = \infty$, respectively.)
The supremum and infimum cannot always be replaced by the maximum and
minimum---in some cases even finite values $\alpha$ and $\beta$ may
not be achieved for any choice of $X\in\A$ and $Y\in\B$.

Semidefinite programs have associated with them a powerful theory of
{\it duality}, which refers to the special relationship between the
primal and dual problems.
The property of {\it weak duality}, which holds for all semidefinite
programs, is stated in the following theorem.

\begin{theorem}[Weak duality]
For every semidefinite program $(\Psi,A,B)$ as defined above, it holds
that $\alpha \leq \beta$.
\end{theorem}

\noindent
This property implies that every dual feasible operator
$Y\in\B$ provides an upper bound of $\ip{B}{Y}$ on the value
$\ip{A}{X}$ that is achievable over all choices of a primal feasible
$X\in\A$, and likewise every primal feasible operator $X\in\A$
provides a lower bound of $\ip{A}{X}$ on the value $\ip{B}{Y}$ that
is achievable over all choices of a dual feasible $Y\in\B$.

It is not always the case that $\alpha = \beta$ for a given
semidefinite program $(\Psi,A,B)$, even when $\alpha$ and $\beta$ are
finite.
For most semidefinite programs that arise in practice, however, it 
is the case that $\alpha = \beta$, which is a situation known as
{\it strong duality}.
There are different conditions under which this property is
guaranteed, one of which is given by the following theorem.

\pagebreak[3]

\begin{theorem}[Slater-type condition for strong duality]
\label{theorem:strong-duality}
The following two implications hold for every semidefinite program
$(\Psi,A,B)$ as defined above.
\begin{mylist}{\parindent}
\item[1.] 
Strict primal feasibility:
If $\beta$ is finite and there exists an operator $X>0$ such that
$\Psi(X) < B$, then $\alpha = \beta$ and there exists $Y\in\B$ such
that $\ip{B}{Y} = \beta$.
\item[2.]
Strict dual feasibility:
If $\alpha$ is finite and there exists an operator $Y>0$ such that
$\Psi^{\ast}(Y) > A$, then $\alpha = \beta$ and there exists $X\in\A$
such that $\ip{A}{X} = \alpha$.
\end{mylist}
\end{theorem}

One may consider a general computational problem that asks for the
optimal primal and dual values of a given semidefinite program,
possibly up to some specified accuracy.
There are various ways in which this may be done, one of which is to
phrase the problem as a {\it promise problem} \cite{EvenS+84} such as
the following one.

\begin{problem} \label{prob:sdp-approximation}
  The {\it semidefinite programming approximation} problem is as
  follows.
  \begin{center}
    \begin{tabular}{lp{5.5in}}
      {\it Input:} &
      A semidefinite program $(\Psi,A,B)$ over $\X = \complex^n$ and
      $\Y=\complex^m$, an accuracy parameter $\varepsilon>0$, and a
      positive integer $R$.\\[1mm]
      {\it Promise:} &
      The set $\A_{\varepsilon}$ is non-empty, and for every $X\in\A$ it
      holds that $\norm{X}_2\leq R$.
      (In the terminology of \cite{GrotschelLS93}, the primal feasible
      region $\A$ of $(\Psi,A,B)$ is {\it well-bounded}, with parameters
      $\varepsilon$ and $R$.)
      \\[1mm]
        {\it Output:} &
        A real number $\gamma$ such that 
        $\abs{\gamma - \alpha} < \varepsilon$, where $\alpha$ is the
        optimal value of the primal problem associated with
        $(\Psi,A,B)$.
    \end{tabular}
  \end{center}
\end{problem}

\noindent
The description of this problem does not explicitly state how the
super-operator $\Psi$ is to be represented, but we will assume it is
specified by the matrix representation of $J(\Psi)$.
Other forms, including Stinespring representations and Kraus
representations, are easily converted to this form.
It is also assumed that the entries of $J(\Psi)$, $A$, and $B$ have
rational real and imaginary parts.

The computational problem stated above can be solved in polynomial
time using the ellipsoid method \cite{GrotschelLS93}, as the following
theorem states.

\begin{theorem}
\label{theorem:ellipsoid-algorithm}
There exists an algorithm that solves the semidefinite programming
approximation problem stated above that runs in time polynomial in
$n$, $m$, $\log(R)$, $\log(1/\varepsilon)$, and the maximum bit-length
of the entries of $J(\Psi)$, $A$, and $B$.
\end{theorem}

\noindent
Here, the {\it bit length} of a complex number $z = (a/b) + i(c/d)$ is
the number of bits needed to represent the 4-tuple $(a,b,c,d)$, where
$a$, $b$, $c$, and $d$ are integers represented in binary.

Note that the above problem asks only for an approximation to the
optimal primal value, but the simple transformation
$(\Psi,A,B) \rightarrow (-\Psi,-B,-A)$ shows that any algorithm for it
also allows one to approximate the optimal dual value.
(Alternately, the ellipsoid method can be applied directly to the dual
problem.)

It is possible to approximate more general classes of semidefinite
programs efficiently.
For instance, the bound $\norm{X}_2\leq R$ need not hold for every
primal feasible $X$, provided certain assumptions are known about the
size of the optimal solution.
These generalizations are not important for this paper, and the above
problem can be more easily fit to the general presentation of
\cite{GrotschelLS93} (which is described in the setting of
semidefinite programming in \cite{Lovasz03}).

It should be noted that one would typically not use the ellipsoid
method to solve semidefinite programming problems in practise, given
that {\it interior point methods} \cite{Alizadeh95,deKlerk02} are
significantly faster.
In strictly formal terms, however, interior point methods have not
been proved to run in polynomial time using the Turing machine model
of computation.

\section{A semidefinite program for the completely bounded trace norm}
\label{sec:SDP}

We will now describe and analyze a semidefinite program whose optimal 
(primal and dual) value is $\triplenorm{\Phi}_1^2$, where
$\Phi\in\trans{\X,\Y}$ is an arbitrary super-operator given by a
Stinespring representation
\[
\Phi(X) = \tr_{\Z}(A X B^{\ast})
\]
for $A,B\in\lin{\X,\Y\otimes\Z}$.
It is assumed further that $\Z$ has the minimal dimension
$\op{dim}(\Z) = \op{rank}(J(\Phi))$ for which such a Stinespring
representation exists.

The primal and dual problems for the semidefinite program we will
consider may be stated informally as follows:
\begin{center}
  \begin{minipage}{3in}
    \centerline{\underline{Primal problem}}\vspace{-7mm}
    \begin{align*}
    \text{maximize:}\quad & \ip{B B^{\ast}}{W}\\
    \text{subject to:}\quad & 
    \tr_{\Y}(W) = \tr_{\Y}\left(A\rho A^{\ast}\right),\\
    & \rho\in\density{\X},\\
    & W\in\pos{\Y\otimes\Z}.
    \end{align*}
  \end{minipage}
  \begin{minipage}{3in}
    \centerline{\underline{Dual problem}}\vspace{-7mm}
    \begin{align*}
    \text{minimize:}\quad & \norm{A^{\ast}(\I_{\Y}\otimes Z) A}_{\infty} \\
    \text{subject to:}\quad 
    & \I_{\Y}\otimes Z \geq B B^{\ast},\\
    & Z \in \pos{\Z}.\\
    \end{align*}
  \end{minipage}
\end{center}

\noindent
These problems are associated with the semidefinite program that is
more formally specified as follows.
We define a Hermiticity-preserving super-operator
\[
\Psi: \lin{\X\oplus(\Y\otimes\Z)} \rightarrow \lin{\complex\oplus\Z}
\]
as
\[
\Psi\begin{pmatrix} X & \cdot \\ \cdot & W\end{pmatrix}
= \begin{pmatrix}\tr(X) & 0\\ 0 & \tr_{\Y}(W - A X
A^{\ast})\end{pmatrix}.
\]
The adjoint super-operator
\[
\Psi^{\ast}:\lin{\complex\oplus\Z}\rightarrow\lin{\X\oplus(\Y\otimes\Z)}
\]
is given by
\[
\Psi^{\ast}\begin{pmatrix} \lambda & \cdot \\ \cdot & Z\end{pmatrix}
=
\begin{pmatrix}
\lambda \I_{\X} - A^{\ast}(\I_{\Y} \otimes Z)A & 0\\
0 & \I_{\Y} \otimes Z
\end{pmatrix}.
\]
(In these expressions of $\Psi$ and $\Psi^{\ast}$, the symbol $\cdot$
denotes an operator or vector of the appropriate dimensions upon which
the output of these super-operators does not depend.)
We also define $C\in\herm{\X\oplus(\Y\otimes\Z)}$ and 
$D\in\herm{\complex\oplus\Z}$ as
\[
C = \begin{pmatrix} 0 & 0\\ 0 & B B^{\ast} \end{pmatrix}
\quad\quad\text{and}\quad\quad
D = \begin{pmatrix} 1 & 0\\ 0 & 0 \end{pmatrix}.
\]

\pagebreak[3]
\noindent
Now, the primal and dual problem associated with $(\Psi,C,D)$ may be
expressed as follows:
\begin{center}
  \begin{minipage}{3in}
    \centerline{\underline{Primal problem}}\vspace{-7mm}
    \begin{align*}
    \text{maximize:}\quad & \ip{B B^{\ast}}{W}\\
    \text{subject to:}\quad & 
    \tr_{\Y}(W) \leq \tr_{\Y}\left(A X A^{\ast}\right),\\
    & \tr(X)\leq 1,\\
    & X\in\pos{\X},\\
    & W\in\pos{\Y\otimes\Z}.
    \end{align*}
  \end{minipage}
  \begin{minipage}{3in}
    \centerline{\underline{Dual problem}}\vspace{-7mm}
    \begin{align*}
    \text{minimize:}\quad & \lambda\\
    \text{subject to:}\quad 
    & \lambda \I_{\X} \geq A^{\ast}(\I_{\Y}\otimes Z) A \\
    & \I_{\Y}\otimes Z \geq B B^{\ast},\\
    & \lambda \geq 0,\\
    & Z \in \pos{\Z}.
    \end{align*}
  \end{minipage}
\end{center}

Notice that for any choice of a primal feasible operator
\begin{equation} \label{eq:block-matrix}
\begin{pmatrix}
X & M\\ M^{\ast} & W
\end{pmatrix},
\end{equation}
there exist operators $P\in\pos{\X}$ and
$Q\in\pos{\Y\otimes\Z}$ such that $\tr(X+P) = 1$
and 
\[
\tr_{\Y}(W+Q) = \tr_{\Y}\left(A (X+ P) A^{\ast}\right).
\]
The operator
\[
\begin{pmatrix}
X + P & M\\ M^{\ast} & W + Q
\end{pmatrix}
\]
is therefore primal feasible, and obtains at least the value achieved 
by \eqref{eq:block-matrix} (by virtue of the fact that $B B^{\ast}$ is
positive semidefinite).
This accounts for the informal statement of the primal problem above,
where the inequality constraints are replaced by equality constraints.
The dual problem above is obviously equivalent to its informal
statement, because $A^{\ast} (\I_{\Y} \otimes Z) A$ is positive
semidefinite for positive semidefinite $Z$, and therefore
\[
\min\left\{ \lambda \geq 0 \,:\,
\lambda\I_{\X} \geq A^{\ast} (\I_{\Y} \otimes Z) A\right\}
= \norm{A^{\ast} (\I_{\Y} \otimes Z) A}_{\infty}.
\]

\subsection*{Strong duality}

We will first verify that strong duality holds for the above
semidefinite program, using Theorem~\ref{theorem:strong-duality}.
First note that it is clear that the optimal primal value $\alpha$ is
finite, for it must hold that $\tr(W) \leq \norm{A}_{\infty}^2$ for
any primal feasible operator with the form \eqref{eq:block-matrix},
and therefore $\alpha \leq \norm{A}_{\infty}^2 \norm{B}_{\infty}^2$.

Now, to verify strict dual feasibility, suppose that $\mu$ and
$\lambda$ are positive real numbers such that
$\mu > \norm{B}_{\infty}^2$ and $\lambda> \mu \norm{A}_{\infty}^2$.
Then
\[
\begin{pmatrix}
\lambda & 0 \\ 0 & \mu \I_{\Z}
\end{pmatrix}
> 0
\quad\quad\text{and}\quad\quad
\Psi^{\ast}
\begin{pmatrix}
\lambda & 0 \\ 0 & \mu \I_{\Z}
\end{pmatrix}
=
\begin{pmatrix}
\lambda \I_{\X} - \mu A^{\ast} A & 0 \\ 0 & \mu \I_{\Y}\otimes\I_{\Z}
\end{pmatrix}
>
\begin{pmatrix}
0 & 0 \\ 0 & B B^{\ast}
\end{pmatrix},
\]
which illustrates strict dual feasibility.
Thus, by Theorem~\ref{theorem:strong-duality}, the optimal value
$\alpha$ associated with the primal problem is equal to the optimal
dual value $\beta$, and is achieved for some choice of a primal
feasible operator.

One may wonder whether the semidefinite program above is also strictly
primal feasible.
Having already established strong duality, it is not really essential
that this is proved, but it may be noted that strict primal
feasibility indeed does hold, relying on the assumption
$\dim(\Z)=\op{rank}(J(\Phi))$.
This observation, which happens to imply that the optimal dual value
is achieved for some dual feasible operator, will follow from the
discussion of computational efficiency below.

\subsection*{Optimal value}

Now let us verify that the optimal value $\alpha = \beta$ of our
semidefinite program is equal to $\triplenorm{\Phi}_1^{2}$.
Define $\W = \complex^k$ for 
$k = \max\{\op{dim}(\X),\,\op{dim}(\Y\otimes\Z)\}$.
Given that $\op{dim}(\W) \geq \op{dim}(\X)$, it holds that
\begin{align*}
\triplenorm{\Phi}^2_1 
& =
\max_{\substack{u,v\in\sphere{\X\otimes\W}\\U\in\unitary{\Y\otimes\W}}}
\abs{\ip{U}{\tr_{\Z}\((A\otimes\I_{\W})u
    v^{\ast}(B^{\ast}\otimes\I_{\W})\)}}^2\\
& =
\max_{\substack{u\in\sphere{\X\otimes\W}\\U\in\unitary{\Y\otimes\W}}}
\norm{
(B^{\ast}\otimes\I_{\W})(U^{\ast}\otimes\I_{\Z})(A\otimes\I_{\W})u}^2\\
& =
\max_{\substack{u\in\sphere{\X\otimes\W}\\U\in\unitary{\Y\otimes\W}}}
u^{\ast}(A^{\ast}\otimes\I_{\W})(U\otimes\I_{\Z})
(B B^{\ast}\otimes\I_{\W})(U^{\ast}\otimes\I_{\Z})(A\otimes\I_{\W}) u\\
& =
\max_{\substack{u\in\sphere{\X\otimes\W}\\U\in\unitary{\Y\otimes\W}}}
\ip{B B^{\ast}}{\tr_{\W}\left[
(U^{\ast}\otimes\I_{\Z})(A\otimes\I_{\W}) u u^{\ast} (A^{\ast}\otimes\I_{\W})
(U\otimes\I_{\Z})\right]}.
\end{align*}

Now define two sets $\Q,\R\subseteq\pos{\Y\otimes\Z}$ as
\begin{align*}
\Q & = 
\left\{W\in\pos{\Y\otimes\Z}\,:\,\tr_{\Y}(W) = \tr_{\Y}(A \rho
A^{\ast})\;\text{for some choice of $\rho\in\density{\X}$}\right\},\\
\R & = \left\{
\tr_{\W}\left[
(U^{\ast}\otimes\I_{\Z})(A\otimes\I_{\W}) u u^{\ast} 
(A^{\ast}\otimes\I_{\W})(U\otimes\I_{\Z})\right]\,:\,
u\in\sphere{\X\otimes\W},\;U\in\unitary{\Y\otimes\W}\right\}.
\end{align*}
Our interest in the set $\R$ is clear, for the equation above has
established that
\[
\triplenorm{\Phi}^2_1 = \max_{W\in\R} \ip{B B^{\ast}}{W}.
\]
The set $\Q$, on the other hand, is of interest because the optimal
value $\alpha$ of the primal problem for the semidefinite program
defined above is given by
\[
\alpha = \max_{W\in\Q} \ip{B B^{\ast}}{W}.
\]
To establish that $\alpha = \triplenorm{\Phi}_1^2$, if therefore
suffices to prove that $\Q = \R$, which is easily done as follows.

First consider an arbitrary choice of $u\in\sphere{\X\otimes\W}$ and
$U\in\unitary{\Y\otimes\W}$, and let
\[
W = \tr_{\W}\left[
(U^{\ast}\otimes\I_{\Z})(A\otimes\I_{\W}) u u^{\ast} (A^{\ast}\otimes\I_{\W})
(U\otimes\I_{\Z})\right].
\]
Then $\tr_{\Y}(W)=\tr_{\Y} \left( A \tr_{\W}(u u^{\ast}) A^{\ast}\right)$,
and so it holds that $W\in\Q$, which proves $\R\subseteq\Q$.

Now consider an arbitrary element $W\in\Q$, and let
$\rho\in\density{\X}$ be a density operator satisfying $\tr_{\Y}(W) =
\tr_{\Y}(A\rho A^{\ast})$.
Given that we have chosen $\W$ to have dimension at least as large as
that of both $\X$ and $\Y\otimes\Z$, there must exist vectors
$u\in\sphere{\X\otimes\W}$ and $w\in\Y\otimes\Z\otimes\W$ such that
$\rho=\tr_{\W}(u u^{\ast})$ and $W = \tr_{\W}(w w^{\ast})$.
This implies that
\[
\tr_{\Y\otimes\W}(w w^{\ast})
= \tr_{\Y\otimes\W}\left((A\otimes\I_{\W})u
u^{\ast}(A^{\ast}\otimes\I_{\W})\right),
\]
so there must exist $U\in\unitary{\Y\otimes\W}$ such that
$(U^{\ast}\otimes\I_{\Z})(A\otimes\I_{\W})u=w$.
Therefore
\[
W = \tr_{\W}(w w^{\ast}) = 
\tr_{\W}\left[
  (U^{\ast}\otimes\I_{\Z})(A\otimes\I_{\W}) u u^{\ast}
  (A^{\ast}\otimes\I_{\W}) (U\otimes\I_{\Z})\right],
\]
which proves that $W\in\R$, so that $\Q\subseteq\R$ as required.

\subsection*{Computational efficiency}

Now let us verify that the optimal value $\triplenorm{\Phi}_1^2$ of
the semidefinite program described above can be approximated by an
efficient computation.
By Theorem~\ref{theorem:ellipsoid-algorithm} our task is to argue
that suitable parameters $R$ and $\varepsilon$ for the promise in
Problem~\ref{prob:sdp-approximation} can be determined.

For the sake of clarity, let us summarize our notation:
we have $\X=\complex^n$, $\Y=\complex^m$, and $\Z = \complex^r$, and
$\Phi\in\trans{\X,\Y}$ is the super-operator given by
\[
\Phi(X) = \tr_{\Z}(A X B^{\ast})
\]
for which we wish to approximate $\triplenorm{\Phi}_1^2$.
The semidefinite program that represents this quantity is represented
by the Hermiticity-preserving super-operator
$\Psi\in\trans{\X\oplus(\Y\otimes\Z),\complex\oplus\Z}$ and
Hermitian operators $C\in\herm{\X\oplus(\Y\otimes\Z)}$ and
$D\in\herm{\complex\oplus\Z}$ as described above.
We will take $N$ to be the total bit-length of this semidefinite
program, which is polynomially related to $n$, $m$ and the maximum
bit-length of the entries of $A$ and $B$.

First, it is clear that every primal feasible operator has trace
bounded by $1 + \norm{A}_{\infty}^2$.
Given that the Frobenius norm is upper-bounded by the trace for
positive semidefinite operators, it therefore suffices to choose
$R = 1 + \norm{A}_{\infty}^2$, which is obviously bounded by
$2^{c N}$ for some positive integer constant $c$.

The specification of $\varepsilon$ is slightly more complicated.
Consider first the operator $\tr_{\Y}(A A^{\ast})$.
We have chosen $\Z$ to have minimal dimension to admit a Stinespring
representation of $\Phi$, and from this assumption it follows that
$\tr_{\Y}(A A^{\ast})$ is positive definite.
Using the assumption that the real and imaginary parts of the entries
of $A$ are rational, along with the fact that nonzero roots of integer
polynomials cannot be too close to zero (see, for instance,
Theorem 2.9 of \cite{Bugeaud04}), one may derive a lower-bound on the
smallest eigenvalue of $\tr_{\Y}(A A^{\ast})$.
For the purposes of this analysis, it suffices to note that there
exists an integer constant $d_0 \geq 1$ such that
for $\delta = 2^{-d_0 N}$ we have that the smallest eigenvalue of
$\tr_{\Y}(A A^{\ast})$ is at least $\delta$, and therefore 
$\delta \, \I_{\Z} \leq \tr_{\Y}(A A^{\ast})$.

Now consider the operator
\[
P = \begin{pmatrix}
X & 0\\ 0& W
\end{pmatrix}
\]
where
\[
X = \frac{3}{4 n} \I_{\X}
\quad\quad\text{and}\quad\quad
W = \frac{3}{8nm}\I_{\Y}
\otimes \tr_{\Y}(A A^{\ast}),
\]
along with any choice of a real number $\varepsilon > 0$ that
satisfies
\[
\varepsilon \leq \frac{\delta}{8 nm}.
\]
Let us note, in particular, that this bound holds for 
$\varepsilon = 2^{-dN}$ for some choice of a positive integer
constant $d$.
It is our goal to show that every Hermitian operator whose distance
from $P$ is at most $\varepsilon$ (with respect to the Frobenius norm)
lies within the primal feasible set $\A$, and therefore that
$\A_{\varepsilon}$ is nonempty.
In other words, for any choice of operators
$H\in\herm{\X}$, $K\in\herm{\Y\otimes\Z}$, and
$M\in\lin{\Y\otimes\Z,\X}$ satisfying
\[
\norm{
\begin{pmatrix}
H & M\\ M^{\ast} & K
\end{pmatrix}
}_2 < \varepsilon,
\]
we wish to prove that
\begin{equation} \label{eq:perturbed}
\begin{pmatrix}
X + H & M\\ M^{\ast} & W + K
\end{pmatrix}
\end{equation}
is primal feasible.

It is clear that $\varepsilon \I < P$, and therefore
\eqref{eq:perturbed} is positive semidefinite.
As $\norm{K}_{\infty}<\varepsilon$ it follows that
\[
W + K \leq 
W + \varepsilon\,\I_{\Y\otimes\Z}
\leq
\frac{1}{2nm} \I_{\Y} \otimes \tr_{\Y}(A A^{\ast})
\]
and therefore
\[
\tr_{\Y}(W + K) \leq \frac{1}{2 n} \tr_{\Y}(A A^{\ast}).
\]
As $\norm{H}_{\infty} \leq \varepsilon$ it holds that
\[
\frac{1}{2 n} \I_{\X} \leq X - \varepsilon \I_{\X} \leq 
X + H
\]
and therefore
\[
\frac{1}{2 n} \tr_{\Y}(A A^{\ast}) \leq \tr_{\Y}(A (X + H) A^{\ast}).
\]
It follows that
$\tr_{\Y}(W + K)\leq \tr_{\Y}(A (X + H) A^{\ast})$
and therefore the above operator \eqref{eq:perturbed} is primal
feasible as required.

We have shown that the requirements of the promise in
Problem~\ref{prob:sdp-approximation} are met for $R = 2^{c N}$ and
$\varepsilon = 2^{-d N}$ for some positive integer constants $c$ and
$d$.
By Theorem~\ref{theorem:ellipsoid-algorithm} the value
$\triplenorm{\Phi}_1^2$ may therefore be approximated to within error
$\varepsilon$ in time polynomial in $n$, $m$ and the size of the
entries of $A$ and $B$.
(It is possible of course to choose a smaller error, 
$\varepsilon = 2^{-p(N)}$ for any polynomial $p$ for instance, if
this is desired.)

\section{A simpler semidefinite program for quantum channel distance}

A somewhat simpler semidefinite program exists for the completely
bounded trace norm of the difference between two quantum channels,
which is a special case that is relevant to quantum information.
This case was discussed in \cite{GilchristLN05}, and shown to reduce
to a convex optimization problem.
The discussion that follows is somewhat different, and is derived from
the {\it refereed quantum games} framework of \cite{GutoskiW07}.

Suppose hereafter in this section that $\Phi = \Phi_0 - \Phi_1$ for
quantum channels $\Phi_0,\Phi_1\in\trans{\X,\Y}$, and consider the
semidefinite program whose primal and dual problems are as follows:
\begin{center}
  \begin{minipage}{2.5in}
    \centerline{\underline{Primal problem}}\vspace{-7mm}
    \begin{align*}
      \text{maximize:}\quad & \ip{J(\Phi)}{W}\\
      \text{subject to:}\quad & W\leq \I_{\Y}\otimes \rho,\\
      & W\in\pos{\Y\otimes\X},\\
      & \rho\in\density{\X}.
    \end{align*}
  \end{minipage}
  \hspace*{12mm}
  \begin{minipage}{2.5in}
    \centerline{\underline{Dual problem}}\vspace{-7mm}
    \begin{align*}
      \text{minimize:}\quad & \norm{\tr_{\Y}(Z)}_{\infty}\\
      \text{subject to:}\quad & Z\geq J(\Phi),\\
      & Z\in\pos{\Y\otimes\X}.\\
    \end{align*}
  \end{minipage}
\end{center}

\noindent
As in the previous section, these problems can be matched to the
formal description of a semidefinite program $(\Psi,C,D)$, for which
strong duality is easily proved.
Our goal will be to prove that the optimal value of this semidefinite
program is given by $\frac{1}{2}\triplenorm{\Phi}_1$.

Given that $\Phi$ is the difference between completely positive
super-operators, it holds \cite{GilchristLN05,RosgenW05,Watrous05} that
\[
\triplenorm{\Phi}_1 
= \max_{u\in\sphere{\X\otimes\X}} 
\norm{(\Phi\otimes\I_{\lin{\X}})(u u^{\ast})}_1.
\]
Given that the operator $(\Phi\otimes\I_{\lin{\X}})(u u^{\ast})$ is
the difference between two density operators for every
$u\in\sphere{\X\otimes\W}$, it follows that
\[
\triplenorm{\Phi}_1
= 2 \max\left\{
\ip{P}{(\Phi\otimes\I_{\lin{\X}})(u u^{\ast})}
\,:\,
u\in\sphere{\X\otimes\X},\; P\in\pos{\Y\otimes\X},\;P\leq
\I_{\Y\otimes\X}
\right\}.
\]

Now, for every unit vector $u\in\X\otimes\X$ there is a corresponding
operator $B\in\lin{\X}$ with $\norm{B}_2 = 1$ such that
\[
u=\sum_{1\leq i,j\leq n} \ip{E_{i,j}}{B} e_i \otimes e_j.
\]
For this choice of $B$ we have
\[
(\I_{\Y}\otimes B) J(\Phi) (\I_{\Y}\otimes B^{\ast})
= (\Phi\otimes\I_{\lin{\X}})(u u^{\ast})
\]
and so
\[
\triplenorm{\Phi}_1
= 2 \max_{B,P}
\ip{(\I\otimes B^{\ast}) P (\I\otimes B)}{J(\Phi)}
\]
where the maximum is over all $B\in\lin{\X}$ with $\norm{B}_2 = 1$ and
$P\in\pos{\Y\otimes\X}$ with $P\leq \I_{\Y\otimes\X}$.

Now define sets $\Q$ and $\R$ as follows:
\begin{align*}
\Q & = \left\{R\in\pos{\Y\otimes\X}\,:\,R \leq \I_{\Y} \otimes
\rho \;\text{for some $\rho\in\density{\X}$}\right\},\\
\R & = \left\{(\I_{\Y}\otimes B^{\ast}) P (\I_{\Y}\otimes B)\,:\,
B\in\lin{\X},\,P\in\pos{\Y\otimes\X},\;
\norm{B}_2 = 1,\; P\leq \I_{\Y\otimes\X}\right\}.
\end{align*}
It holds that
\[
\triplenorm{\Phi}_1 = 2 \sup_{X\in\R} \ip{J(\Phi)}{X}
\]
while the optimal value of the semidefinite program is
\[
\alpha = \sup_{X\in\Q} \ip{J(\Phi)}{X}.
\]
The fact that $\alpha = \frac{1}{2}\triplenorm{\Phi}_1$ therefore
follows from the equality $\Q=\R$, which is easily proved by selecting
$\rho$ or $B$ so that $\rho = B^{\ast} B$.

\section{Connections with known results}

This section describes two interesting connections between the
semidefinite programming formulation from Section~\ref{sec:SDP} and
known results, the first being directly about completely bounded
norms, and the second concerning the fidelity function.

\subsection*{Spectral norms of Stinespring representations}

The following theorem gives an alternate characterization of the
completely bounded trace norm (or diamond norm).
Proofs can be found in Kitaev, Shen and Vyalyi \cite{KitaevSV02} and
Paulsen \cite{Paulsen02}.
The two proofs use rather different techniques, and here the theorem
is proved in a third way using semidefinite programming duality.

\begin{theorem}
For every super-operator $\Phi\in\trans{\X,\Y}$, it holds that
\begin{equation} \label{eq:diamond-inf}
\triplenorm{\Phi}_1 = \inf_{(A,B)} \norm{A}_{\infty}\norm{B}_{\infty},
\end{equation}
where the infimum is over all Stinespring pairs $(A,B)$ for $\Phi$.
\end{theorem}

\begin{proof}
For any Stinespring pair $(A,B)$ of $\Phi$, where
$A,B\in\lin{\X,\Y\otimes\Z}$, and for any choice of $\W = \complex^k$,
it holds that
\begin{align*}
\norm{(\Phi\otimes\I_{\lin{\W}})(X)}_1 
& = \norm{\tr_{\Z} \left[ (A\otimes\I_{\W}) X 
(B^{\ast}\otimes\I_{\W})\right]}_1\\
& \leq 
\norm{(A\otimes\I_{\W}) X (B^{\ast}\otimes\I_{\W})}_1 \\
& \leq
\norm{A}_{\infty} \norm{X}_1 \norm{B}_{\infty}
\end{align*}
for all $X\in\lin{\X\otimes\W}$.
It follows that $\triplenorm{\Phi}_1\leq\norm{A}_{\infty}\norm{B}_{\infty}$.

To prove that the infimum is no larger than $\triplenorm{\Phi}_1$,
first choose an arbitrary Stinespring pair $(A,B)$ of $\Phi$, where
$A,B\in\lin{\X,\Y\otimes\Z}$.
The optimal value for the dual problem stated in Section~\ref{sec:SDP}
does not change if $Z$ is restricted to be positive definite, provided
we accept that an optimal solution may not be achieved.
We therefore have
\[
\triplenorm{\Phi}_1^2 = 
\inf\{\norm{A^{\ast}(\I_{\Y}\otimes Z) A}_{\infty}\,:\,
\I_{\Y}\otimes Z \geq B B^{\ast},\;
Z \in \pd{\Z}\}.
\]
Thus, for a given $\varepsilon > 0$, we may choose $Z\in\pd{\Z}$ such
that
\[
\norm{A^{\ast}(\I_{\Y}\otimes Z) A}_{\infty}
\leq \left(\triplenorm{\Phi}_1 + \varepsilon\right)^2
\]
and $\I_{\Y}\otimes Z \geq B B^{\ast}$.
This second inequality is equivalent to
\[
\norm{
\left(\I_{\Y}\otimes Z^{-1/2}\right)B B^{\ast}
\left(\I_{\Y}\otimes Z^{-1/2}\right)}_{\infty}\leq 1.
\]
So now we have that
\[
\norm{\left(\I_{\Y}\otimes Z^{1/2}\right)A}_{\infty}\,
\norm{\left(\I_{\Y}\otimes Z^{-1/2}\right)B}_{\infty}
\leq \triplenorm{\Phi}_1 + \varepsilon,
\]
and it holds that
\[
\left(
\left(\I_{\Y}\otimes Z^{1/2}\right)A,\,
\left(\I_{\Y}\otimes Z^{-1/2}\right)B\right)
\]
is a Stinespring pair for $\Phi$.
This establishes that the infimum equals $\triplenorm{\Phi}_1$ in the
expression \eqref{eq:diamond-inf}, which completes the proof.
\end{proof}

\subsection*{Connection with fidelity}

Consider the semidefinite program from Section~\ref{sec:SDP}, for the
special case where $\X = \complex$.
Replacing $A$ and $B$ with vectors $u,v\in\Y\otimes\Z$, and making
simplifications, the problems become as follows:
\begin{center}
  \begin{minipage}{3in}
    \centerline{\underline{Primal problem}}\vspace{-7mm}
    \begin{align*}
    \text{maximize:}\quad & \ip{v v^{\ast}}{W}\\
    \text{subject to:}\quad & 
    \tr_{\Y}(W) \leq \tr_{\Y}\(u u^{\ast}\),\\
    & W\in\pos{\Y\otimes\Z}.
    \end{align*}
  \end{minipage}
  \begin{minipage}{3in}
    \centerline{\underline{Dual problem}}\vspace{-7mm}
    \begin{align*}
    \text{minimize:}\quad & \ip{\tr_{\Y}(u u^{\ast})}{Z} \\
    \text{subject to:}\quad 
    & \I_{\Y}\otimes Z \geq v v^{\ast},\\
    & Z \in \pos{\Z}.
    \end{align*}
  \end{minipage}
\end{center}

The quantity that is represented by the optimal value of these
problems is given by the {\it fidelity} function, which is defined as
\[
\fid(P,Q) = \norm{\sqrt{P}\sqrt{Q}}_1 = \tr\sqrt{\sqrt{P} Q \sqrt{P}}
\]
for positive semidefinite operators $P$ and $Q$.
In particular, the optimal value (for the primal and dual problems) is
\begin{equation} \label{eq:fidelity1}
\fid \( \tr_{\Y}(u u^{\ast}) , \tr_{\Y}(v v^{\ast}) \)^2,
\end{equation}
as is now explained.

First, the optimal value of the primal problem follows
from Uhlmann's Theorem \cite{Uhlmann76}, which is as follows.

\begin{theorem}[Uhlmann's Theorem] \label{theorem:Uhlmann}
Let $\Y$ and $\Z$ be finite-dimensional complex vector spaces, and
let $P,Q\in\pos{\Z}$ be positive semidefinite operators, both having
rank at most $\op{dim}(\Y)$.
Then for any choice of $v\in\Y\otimes\Z$ satisfying
$\tr_{\Y}(v v^{\ast}) = Q$, it holds that
\[
\fid(P,Q) = \max\left\{ \abs{\ip{u}{v}}\,:\,
u\in\Y\otimes\Z,\;\tr_{\Y}(u u^{\ast}) = P\right\}.
\]
\end{theorem}

\noindent
It is straightforward to obtain from this theorem (along with simple
properties of the fidelity) the following corollary, which is
precisely the statement that the optimal primal value of our
semidefinite program is given by the fidelity.

\begin{cor}
Assume $u,v\in\Y\otimes\Z$ are vectors, and let
$P = \tr_{\Y}(u u^{\ast})$ and 
$Q = \tr_{\Y}(v v^{\ast})$.
Then
\[
\fid(P,Q)^2 = \op{max}\left\{\ip{v v^{\ast}}{W}\,:\,
W\in\pos{\Y\otimes\Z},\: \tr_{\Y}(W) \leq P\right\}.
\]
\end{cor}

The optimal value of the dual problem is, of course, equal to 
\eqref{eq:fidelity1} by strong duality.
A different way to evaluate the optimal dual value begins with the
following simple proposition.

\begin{prop}
For any vector $v\in\Y\otimes\Z$ and any positive definite operator 
$Z\in \pd{\Z}$ it holds that $\I_{\Y}\otimes Z \geq v v^{\ast}$
if and only if $\ip{\tr_{\Y}(v v^{\ast})}{Z^{-1}} \leq 1$.
\end{prop}

\begin{proof}
It holds that $\I_{\Y}\otimes Z \geq v v^{\ast}$ if and only if
\begin{equation} \label{eq:Alberti-2}
\(\I_{\Y} \otimes Z^{-1/2}\) v v^{\ast}\(\I_{\Y} \otimes Z^{-1/2}\)
\leq \I_{\Y\otimes\Z}.
\end{equation}
Given that the operator on the left-hand-side of \eqref{eq:Alberti-2}
is positive semidefinite and has rank equal to~1, we have that
\eqref{eq:Alberti-2} is equivalent to
\[
\norm{\(\I_{\Y} \otimes Z^{-1/2}\)v} \leq 1,
\]
which in turn is equivalent to
\[
\tr\( \(\I_{\Y} \otimes Z^{-1/2}\)v v^{\ast} \(\I_{\Y} \otimes
Z^{-1/2}\)\)
\leq 1.
\]
As
\[
\tr\( \(\I_{\Y} \otimes Z^{-1/2}\)v v^{\ast} \(\I_{\Y} \otimes
Z^{-1/2}\)\)
=
\ip{\tr_{\Y}(v v^{\ast})}{Z^{-1}},
\]
the proof is complete.
\end{proof}

\noindent
We have that the optimal value of the dual problem does not
change if $Z$ is optimized over only positive definite rather than
positive semidefinite operators (again accepting that the optimal
value may not be achieved for such an operator).
Combined with the proposition just proved, we find that the optimal
dual value is given by
\[
\beta = \inf \left\{ \ip{\tr_{\Y}(u u^{\ast})}{Z}\,:\,
Z\in\pd{\Z},\; \sip{\tr_{\Y}(v v^{\ast})}{Z^{-1}}\leq 1\right\}.
\]
That this value is given by \eqref{eq:fidelity1} follows from
a different characterization of the fidelity due to Alberti
\cite{Alberti83}.

\begin{theorem}[Alberti's Theorem] \label{theorem:Alberti}
  Let $P,Q\in\pos{\Z}$ be positive semidefinite operators.
  Then
  \[
  \(\fid(P,Q)\)^2 = \inf_{Z\in\pd{\Z}} \ip{P}{Z} \sip{Q}{Z^{-1}}.
  \]
\end{theorem}

\noindent
We have therefore established a simple and precise sense in which
Uhlmann's Theorem and Alberti's Theorem are dual statements in finite
dimensions, each implying the other.

\subsection*{Acknowledgments}

Many of the key ideas presented in this paper are contained in
the papers \cite{KitaevW00} and \cite{GutoskiW07}.
I thank Alexei Kitaev and Gus Gutoski for discussions during
and after the collaborations that produced these papers, which helped
to clarify these ideas.
I also thank Bill Rosgen and Mary Beth Ruskai for helpful comments.
This research was supported by Canada's NSERC and the Canadian
Institute for Advanced Research (CIFAR).

\bibliographystyle{alpha}

\newcommand{\etalchar}[1]{$^{#1}$}

\end{document}